\documentclass[journal]{IEEEtran}

\usepackage{amssymb}
\usepackage{amsmath}
\usepackage{cite}
\usepackage{url}
\usepackage{empheq}
\usepackage{xcolor}
\usepackage{graphicx}
\usepackage{subfigure}
\usepackage{enumitem}
\usepackage{fancyhdr}
\usepackage{mdwmath}	
\usepackage{mdwtab}
\usepackage{caption}
\usepackage{amsthm}
\usepackage{algorithm}
\usepackage{algorithmic}

\newtheorem{lemma}{Lemma}
\newtheorem{remark}{Remark}
\newtheorem{theorem}{Theorem}
\newtheorem{corollary}{Corollary}
\newtheorem{assumption}{Assumption}

\newtheorem{proposition}{Proposition}

\newcommand{\fref}[1]{Fig.~\ref{#1}}

\newcommand*{\QEDA}{\null\nobreak\hfill\ensuremath{\blacksquare}}

\begin{document}
\title{Pinching-Antenna Systems: From Antenna Placement to Antenna Roaming}
\author{Kaidi~Wang,~\IEEEmembership{Member,~IEEE,}
Daniel~K.~C.~So,~\IEEEmembership{Senior~Member,~IEEE,}
and~Zhiguo~Ding,~\IEEEmembership{Fellow,~IEEE}
\thanks{Kaidi~Wang and Daniel~K.~C.~So are with the Department of Electrical and Electronic Engineering, the University of Manchester, Manchester, M1 9BB, UK (email: kaidi.wang@ieee.org; d.so@manchester.ac.uk).}
\thanks{Zhiguo~Ding is with the School of Electrical and Electronic Engineering (EEE), Nanyang Technological University, Singapore 639798 (e-mail: zhiguo.ding@ntu.edu.sg).}}
\maketitle
%%%%%%%%%%%%%%%%%%%%%%%%%%%%%%%%%%%%%%%%%%%%%%%%%
%%%%%%%%%%%%%%%%%%%%%%%%%%%%%%%%%%%%%%%%%%%%%%%%%
\begin{abstract}
This paper investigates pinching-antenna systems with finite antenna movement speed, under which conventional antenna placement is subject to non-negligible repositioning delay, resulting in a fundamental tradeoff between channel quality and effective transmission time. In this context, antenna roaming is proposed as a novel operation mode, in which the antenna moves continuously along the waveguide while simultaneously serving users. By incorporating communication and antenna movement into the same transmission cycle, a unified cycle-duration framework is established to facilitate a fair comparison between antenna placement and antenna roaming. The sum-rate difference is then analytically derived, indicating that antenna roaming avoids dedicated positioning overhead with a loss in channel quality due to spatial averaging. The corresponding sum rate maximization problems are formulated for the two operation modes. The continuous optimization problems are transformed into finite-state sequential decision problems and solved via dynamic programming (DP) based algorithms. For antenna roaming, the optimal unconstrained service-interval partition is analytically characterized, with each antenna position assigned to the user achieving the highest instantaneous rate. Simulation results validate the theoretical analysis, demonstrate the performance advantage of antenna roaming over antenna placement, especially when positioning overhead is significant, and confirm the effectiveness of the DP based solutions in improving the achievable sum rate.
\end{abstract}
\begin{IEEEkeywords}
Pinching-antenna systems, finite antenna movement speed, antenna placement, antenna roaming, sum rate maximization, dynamic programming.
\end{IEEEkeywords}. 
%%%%%%%%%%%%%%%%%%%%%%%%%%%%%%%%%%%%%%%%%%%%%%%%% 
%%%%%%%%%%%%%%%%%%%%%%%%%%%%%%%%%%%%%%%%%%%%%%%%%
\section{Introduction}
The development of sixth-generation (6G) wireless networks calls for substantial advances in spectral efficiency, reliability, and coverage across diverse communication scenarios, prompting extensive research on flexible reconfigurable antenna technologies. To address these requirements, a variety of antenna architectures have been explored, among which reconfigurable intelligent surfaces (RISs), fluid antennas, and movable antennas have received particular attention \cite{wu2019irs, wong2020fluid, zhu2024movable}. While RISs reshape the wireless environment by controlling signal reflections through an array of passive elements \cite{pan2021ris}, the associated cascaded transmitter-RIS-receiver link incurs additional path loss, which may weaken the received signal and restrict the achievable performance. In contrast, fluid antennas and movable antennas enable antenna position reconfiguration through various mechanisms, including fluidic and electromechanical implementations \cite{wu2024fluid, ma2024movable}. However, their available movement regions are generally confined to several wavelengths, allowing them to exploit small-scale channel variations while providing limited access to large-scale spatial diversity. These shortcomings underscore the need for alternative antenna architectures that offer broader spatial reconfigurability and more effective adaptation to heterogeneous wireless environments.

Recently, pinching antennas have emerged as a promising flexible reconfigurable antenna technology. Specifically, in pinching-antenna systems, dielectric particles are attached to waveguides to create radiation points that couple guided signals into free space at desired positions \cite{suzuki2022pinching, ding2024pin}. By transmitting signals from the feed point to these radiation points through low-loss dielectric waveguides, pinching-antenna systems replace a substantial portion of free-space propagation with guided transmission, thereby shortening the wireless links and mitigating the overall propagation loss \cite{liu2025pinching, kaidi2025pin, xu2025pin}. Moreover, the radiation points can be flexibly adjusted along the waveguides to exploit large-scale spatial variations, facilitating the establishment of line-of-sight (LoS) links by circumventing physical blockages while suppressing interference by leveraging blockage effects on undesired links \cite{ding2025edma, kaidi2025pin4}. In addition to these communication benefits, pinching-antenna systems also provide practical advantages in implementation complexity and cost, owing to their simple dielectric waveguides and compact radiating elements \cite{yang2025pinching, lu2026pin}.
%%%%%%%%%%%%%%%%%%%%%%%%%%%%%%%%%%%%%%%%%%%%%%%%%
\subsection{Related Works}
Existing pinching-antenna systems can be broadly classified into continuous and discrete architectures. Specifically, continuous architectures allow radiation points to be continuously adjusted along the waveguide, providing high spatial flexibility and enabling antenna placement optimization. For a single-waveguide orthogonal multiple-access (OMA) system, \cite{chen2026pin} considered antenna placement under a practical propagation model that accounts for signal attenuation along the waveguide, revealing the tradeoff between in-waveguide power loss and free-space path loss. Within the same architecture, \cite{xu2025pin2} investigated antenna placement for a downlink non-orthogonal multiple access (NOMA) system, with antenna positions and power allocation jointly designed under heterogeneous quality-of-service (QoS) requirements. For multi-waveguide NOMA systems, the authors of \cite{xue2026pin} incorporated antenna placement, power allocation, and user scheduling into a unified resource-allocation framework to exploit the additional spatial degrees of freedom. For multicast transmission, \cite{chen2025pin, shan2026pin} focused on user fairness by formulating minimum rate maximization problems involving antenna placement and transmit beamforming under single- and multiple-waveguide architectures. Moreover, in \cite{zhou2026pin2}, antenna positions and transmit beamforming were optimized for weighted sum rate maximization, for which a gradient based meta-learning framework was developed to improve optimization efficiency.

In parallel, discrete pinching-antenna architectures enable spatial reconfiguration through antenna activation at discrete positions, which can be realized by either activating antennas pre-installed at fixed positions or repositioning antennas among predefined candidate positions. For NOMA assisted systems, \cite{kaidi2025pin2} investigated discrete antenna activation with waveguide assignment, successive interference cancellation (SIC) decoding order, and power allocation for sum rate maximization. Building on this framework, \cite{xu2026pass} considered transmit power minimization by optimizing antenna activation, transmit beamforming, and the number of activated antennas under an adjustable power radiation model. The framework was further extended to integrated sensing and communication (ISAC) in \cite{yu2026pin}, with discrete activation and beamforming optimized under communication constraints. Moreover, \cite{kaidi2025generalized} generalized the discrete pinching-antenna concept to leaky coaxial cables for low-frequency operation, jointly considering radiation-slot activation, user assignment, and power allocation.
%%%%%%%%%%%%%%%%%%%%%%%%%%%%%%%%%%%%%%%%%%%%%%%%%
\subsection{Motivation and Contributions}
Despite the practical advantages of discrete architectures, continuous pinching-antenna systems remain attractive due to the potential performance gains from higher spatial flexibility. However, existing works on continuous architectures typically assume instantaneous antenna placement and therefore neglect the repositioning delay caused by finite antenna movement speed \cite{chen2026pin, xu2025pin2, xue2026pin, chen2025pin, shan2026pin, zhou2026pin2}. In practice, antenna repositioning requires a non-negligible portion of the transmission cycle \cite{liu2026energy, xu2026pin2}, which introduces a fundamental tradeoff between channel quality improvement and effective transmission time. Moreover, since each antenna position determines the repositioning distance to subsequent service points, successive placement decisions become temporally coupled, substantially increasing optimization complexity. From an implementation perspective, antenna placement also imposes stringent hardware requirements, as the antenna must move rapidly and stop precisely at each optimized position \cite{gan2026pin2}. These limitations motivate a new operation mode, termed antenna roaming, in which the antenna moves continuously along the waveguide at a moderate speed while simultaneously providing communication service. To enable a fair comparison, antenna placement and antenna roaming are evaluated over a common transmission cycle that explicitly accounts for both antenna movement and communication time, providing a basis for subsequent performance analysis and optimization.

The main contributions can be summarized as follows:
\begin{itemize}[leftmargin=*]
\item By incorporating a finite antenna movement speed, a practical antenna placement model is established. The model explicitly characterizes the sequential repositioning of the pinching antenna between user-specific service points and includes the associated movement time within each scheduling cycle. The connection between the proposed model and conventional infinite-speed antenna placement is clarified, and the resulting changes in the placement strategy under finite-speed movement are identified.
\item A novel antenna roaming scheme is proposed, where the pinching antenna continuously moves along the waveguide while serving users over dedicated service intervals. A corresponding transmission model is developed based on service-interval partitioning, enabling uninterrupted data transmission throughout each scheduling cycle. The performance advantages and limitations of antenna roaming are clarified through comparison with antenna placement.
\item A unified analytical framework is developed to compare antenna placement and antenna roaming. The sum-rate difference is decomposed into the gain associated with effective transmission time and the loss associated with spatially averaged channel quality, thereby revealing the fundamental tradeoff between the two operation modes. Based on these analytical insights, corresponding QoS constrained sum rate maximization problems are formulated.
\item The structural properties of the formulated optimization problems are investigated. By exploiting the additive chain structure, the original problems are reformulated as finite-state sequential decision problems, for which dynamic programming (DP) based solutions are developed. For antenna roaming, the single-crossing property of the instantaneous rate functions is established, revealing the structure of the optimal unconstrained service-interval partition.
\end{itemize}
Simulation results validate the theoretical analysis and illustrate the distinct operating mechanisms of antenna placement and antenna roaming. The results further demonstrate the performance advantage of antenna roaming under practical finite-speed antenna movement, particularly when repositioning overhead is significant. The effectiveness of the proposed DP based algorithms is also confirmed through comparisons with benchmark spatial service designs.
%%%%%%%%%%%%%%%%%%%%%%%%%%%%%%%%%%%%%%%%%%%%%%%%%
\subsection{Organization}
The remainder of this paper is organized as follows. Section II develops the antenna placement and antenna roaming models, followed by a comparative performance analysis. Section III formulates the corresponding sum rate maximization problems. Sections IV and V develop DP based solutions for service-point deployment and service-interval partitioning, respectively. Section VI presents the simulation results, and Section VII concludes the paper.
%%%%%%%%%%%%%%%%%%%%%%%%%%%%%%%%%%%%%%%%%%%%%%%%%
%%%%%%%%%%%%%%%%%%%%%%%%%%%%%%%%%%%%%%%%%%%%%%%%%
\section{System Model}
Consider a downlink pinching-antenna system, where one base station (BS) serves $N$ single-antenna users via a pinching antenna deployed on a waveguide. The service region is a rectangular area of size $D_x \times D_y$, with the origin located at the center of the region. The waveguide is deployed at height $d$ along the $x$-axis at $y=0$, with the feed point located at the left end of the waveguide, i.e., $\boldsymbol{\psi}^\mathrm{feed}=(-D_x/2,0,d)$. The collection of all users is denoted by $\mathcal{N}=\{1,2,\dots,N\}$, and the position of user $n$ is given by $\boldsymbol{\psi}_n=(x_n,y_n,0)$. Without loss of generality, the users are indexed according to their $x$-coordinates\footnote{The user ordering is adopted to facilitate the description of the antenna movement and service procedure. For users without a predefined order, the service sequence can be optimized through user scheduling.}, i.e.,
\begin{equation}
-\frac{D_x}{2}\le x_1 \le x_2\le \cdots\le x_N\le \frac{D_x}{2}.
\end{equation}
This paper considers a representative scheduling cycle of duration $T$, during which the pinching antenna moves from the left end of the waveguide to the right end. The left-to-right movement is adopted for notational convenience, while the proposed formulation can be equally applied to the reverse direction or to a prescribed subregion of the waveguide.
%%%%%%%%%%%%%%%%%%%%%%%%%%%%%%%%%%%%%%%%%%%%%%%%%
\subsection{Channel and Signal Model}
In the considered pinching-antenna system, the antenna position is represented as a function of its horizontal coordinate $x$, i.e., $\boldsymbol{\psi}(x)=(x,0,d)$, where $x\in[-D_x/2,D_x/2]$. Accordingly, the channel between the pinching antenna and user $n$ is given by
\begin{equation}\label{channel}
h_n(x)=\frac{\eta e^{-j\left(\frac{2\pi}{\lambda}\left\|\boldsymbol{\psi}_n-\boldsymbol{\psi}(x)\right\|+\frac{2\pi}{\lambda_g}\left\|\boldsymbol{\psi}^\mathrm{feed}-\boldsymbol{\psi}(x)\right\|\right)}}{\left\|\boldsymbol{\psi}_n-\boldsymbol{\psi}(x)\right\|},
\end{equation}
where $\eta=\frac{c}{4\pi f_c}$ is the reference channel gain, $c$ is the speed of light, $f_c$ is the carrier frequency, $\lambda$ is the carrier wavelength, $\lambda_g=\lambda/n_\mathrm{eff}$ is the guided wavelength in the dielectric waveguide, $n_\mathrm{eff}$ is the effective refractive index, and $\|\cdot\|$ denotes the Euclidean norm.

For downlink transmission, time division multiple access (TDMA) is adopted, such that only one user is scheduled at any given time and no inter-user interference is present. When user $n$ is served by the pinching antenna at horizontal coordinate $x$, the received signal is given by
\begin{equation}
y_n(x)=h_n(x)\sqrt{P_t}s_n+z_n,
\end{equation}
where $P_t$ is the transmit power at the BS, $s_n$ is the transmitted signal for user $n$, and $z_n$ is the additive white Gaussian noise (AWGN). The corresponding instantaneous achievable rate of user $n$ can be expressed as follows:
\begin{equation}\label{rate}
r_n(x)=\log_2(1+\rho|h_n(x)|^2),
\end{equation}
where $\rho=P_t/\sigma^2$ is the transmit signal-to-noise ratio (SNR), and $\sigma^2$ is the noise power.

\begin{figure}[!t]
\centering\includegraphics[width=88mm]{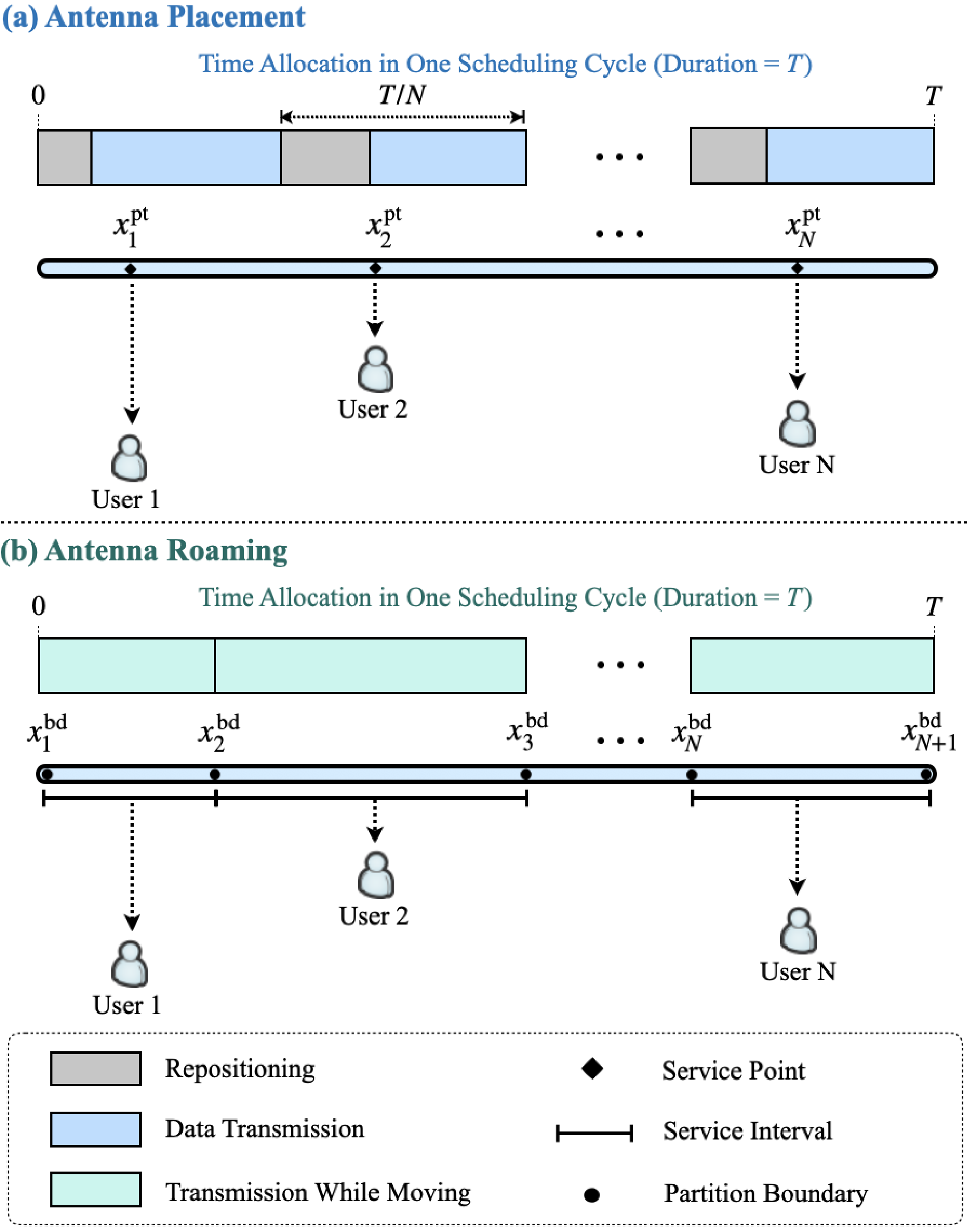}
\caption{System illustration of antenna placement and antenna roaming. In (a), the pinching antenna serves user $n$ at the service point $x_n^\mathrm{pt}$. In (b), the pinching antenna serves user $n$ over the service interval $[x_n^\mathrm{bd},x_{n+1}^\mathrm{bd}]$.}\vspace{-4mm}
\label{system}
\end{figure}
%%%%%%%%%%%%%%%%%%%%%%%%%%%%%%%%%%%%%%%%%%%%%%%%%
\subsection{Antenna Placement Model}
For continuously adjustable pinching antennas, a common operation mode is to reposition the antenna to different service points for serving multiple users \cite{xie2025pin, ding2025analytical}, with data transmission performed after each service point is reached, as shown in \fref{system}(a). In this subsection, the positioning time required to reach each service point is incorporated into the antenna placement model. Following the commonly adopted equal-resource allocation in OMA based systems \cite{ding2024pin}, the scheduling cycle of duration $T$ is equally divided among the $N$ users. Accordingly, each user is assigned a time slot of duration $T/N$ for both antenna positioning and data transmission.

For user $n$, the antenna position $\boldsymbol{\psi}(x_n^\mathrm{pt})$	 is uniquely determined by $x_n^\mathrm{pt}$, which is therefore referred to as the service point for notational simplicity. Under the considered left-to-right scheduling cycle, the service points satisfy
\begin{equation}
x_{n-1}^\mathrm{pt}\le x_n^\mathrm{pt}, \quad \forall n\in\mathcal{N},
\end{equation}
where $x_0^\mathrm{pt}=-D_x/2$ corresponds to the initial horizontal coordinate of the pinching antenna. To minimize the repositioning delay and thereby maximize the data transmission time within each time slot, the pinching antenna is assumed to move at its maximum speed\footnote{To facilitate a tractable characterization of the repositioning delay, the antenna repositioning process is modeled as constant-speed motion, with the speed variations during initial acceleration and final deceleration neglected.} $v_\mathrm{max}$. Accordingly, the time required to reposition the antenna from $x_{n-1}^\mathrm{pt}$ to $x_n^\mathrm{pt}$ is given by
\begin{equation}
\tau_n^\mathrm{po}=\frac{x_n^\mathrm{pt}-x_{n-1}^\mathrm{pt}}{v_\mathrm{max}}.
\end{equation}
Hence, the available transmission time for user $n$ is
\begin{equation}
\tau_n^\mathrm{tx}=\frac{T}{N}-\tau_n^\mathrm{po}.
\end{equation}
To ensure that the pinching antenna can reach each service point within the allocated time slot, the following feasibility condition must hold:
\begin{equation}
\frac{x_n^\mathrm{pt}-x_{n-1}^\mathrm{pt}}{v_\mathrm{max}}\le\frac{T}{N}, \quad \forall n\in\mathcal{N}.
\end{equation}
Based on the instantaneous achievable rate in \eqref{rate}, the average rate of user $n$ under antenna placement can be expressed as a function of two consecutive service points, as follows:
\begin{equation}
R_n^\mathrm{AP}(x_{n-1}^\mathrm{pt},x_n^\mathrm{pt})=\frac{\tau_n^\mathrm{tx}}{T}r_n(x_n^\mathrm{pt})=\left(\!\frac{1}{N}\!-\!\frac{x_n^\mathrm{pt}\!-\!x_{n-1}^\mathrm{pt}}{T v_\mathrm{max}}\!\right)\!r_n(x_n^\mathrm{pt}).
\end{equation}

The relationship between the considered antenna placement model with finite positioning time and conventional antenna placement is clarified in the following remarks.
\begin{remark}
The considered antenna placement model generalizes conventional antenna placement by incorporating the positioning time between consecutive service points. Specifically, conventional antenna placement can be regarded as a special case of the considered model when $v_\mathrm{max}\to\infty$, under which $\tau_n^\mathrm{po}\to 0$ and $R_n^\mathrm{AP}(x_{n-1}^\mathrm{pt},x_n^\mathrm{pt})\to r_n(x_n^\mathrm{pt})/N$. Therefore, this model captures the movement induced performance loss while remaining consistent with conventional antenna placement.
\end{remark}
\begin{remark}
In the antenna placement model, incorporating the positioning time introduces sequential coupling among the service-point decisions. With instantaneous repositioning, the average rate of each user depends solely on its own service point. In contrast, with finite positioning speed, it also depends on the preceding service point, and each selected service point influences the service of subsequent users.
\end{remark}
%%%%%%%%%%%%%%%%%%%%%%%%%%%%%%%%%%%%%%%%%%%%%%%%%
\subsection{Antenna Roaming Model}
Antenna placement requires the pinching antenna to be accurately positioned and stabilized at each optimized service point before data transmission, which may impose stringent mechanical control requirements. To alleviate this requirement, a novel operation mode, referred to as antenna roaming, is proposed, in which the pinching antenna moves continuously along the waveguide without stopping\footnote{Antenna roaming can be realized with a movable dielectric structure that slides along the waveguide while maintaining the physical contact or coupling condition, thereby preserving signal radiation throughout the movement.}. Unlike antenna placement, where each user is served at a specific service point, antenna roaming serves each user over a service interval along the antenna trajectory, as shown in \fref{system}(b).

Under antenna roaming, given the scheduling cycle duration $T$ and the waveguide length $D_x$, the pinching antenna moves at the following constant speed:
\begin{equation}\label{speed}
v_\mathrm{pin}=\frac{D_x}{T}.
\end{equation}
To ensure the feasibility of antenna roaming, the required movement speed must satisfy $v_\mathrm{pin}\le v_\mathrm{max}$. For the considered left-to-right movement, the horizontal coordinate of the pinching antenna at time $t\in[0,T]$ is given by
\begin{equation}\label{position}
x(t)=-\frac{D_x}{2}+v_\mathrm{pin}t.
\end{equation}
Since the antenna moves at a constant speed, its position is uniquely determined by time according to \eqref{position}. As the channel in \eqref{channel} depends on the antenna position, each user's channel becomes time-varying through the antenna trajectory $x(t)$. Under antenna roaming, the antenna trajectory is partitioned into $N$ consecutive service intervals\footnote{The consecutive service-interval structure is adopted for average rate maximization, such that the pinching antenna transmits continuously throughout the antenna trajectory. For other design objectives, such as power consumption minimization or energy efficiency maximization, a nonconsecutive service-interval partition may be considered, allowing the antenna to remain idle over certain portions of the trajectory.}, with the partition boundaries denoted by $\{x_n^\mathrm{bd}\}_{n=1}^{N+1}$ and satisfying
\begin{equation}
-\frac{D_x}{2}=x_1^\mathrm{bd}\le x_2^\mathrm{bd}\le\cdots\le x_{N+1}^\mathrm{bd}=\frac{D_x}{2}.
\end{equation}
Specifically, user $n$ is served over the interval $[x_n^\mathrm{bd},x_{n+1}^\mathrm{bd}]$, during which the pinching antenna moves from $x_n^\mathrm{bd}$ to $x_{n+1}^\mathrm{bd}$ over the corresponding time interval $[t_n^\mathrm{bd},t_{n+1}^\mathrm{bd}]$. The average rate of user $n$ under antenna roaming is given by
\begin{equation}
R_n^\mathrm{AR}=\frac{1}{T}\int_{t_n^\mathrm{bd}}^{t_{n+1}^\mathrm{bd}}\!\!r_n(x(t))\,dt.
\end{equation}
Unlike antenna placement, where users are assigned equal-duration time slots, antenna roaming also optimizes user service durations through the corresponding service-interval lengths. This additional degree of freedom enables more flexible resource allocation and may lead to improved system performance. Since $dx=v_\mathrm{pin}dt$, $x(t_n^\mathrm{bd})=x_n^\mathrm{bd}$, and $x(t_{n+1}^\mathrm{bd})=x_{n+1}^\mathrm{bd}$, the time domain expression can be transformed into the spatial domain as follows:
\begin{align}
R_n^\mathrm{AR}(x_n^\mathrm{bd},x_{n+1}^\mathrm{bd})&=\frac{1}{T v_\mathrm{pin}}\int_{x_n^\mathrm{bd}}^{x_{n+1}^\mathrm{bd}}\!\!r_n(x)\,dx\nonumber\\
&=\frac{1}{D_x}\int_{x_n^\mathrm{bd}}^{x_{n+1}^\mathrm{bd}}\!\log_2(1\!+\!\rho|h_n(x)|^2)\,dx.
\label{arrate}
\end{align}
The integral in \eqref{arrate} represents the spatial accumulation of the instantaneous achievable rate over the service interval\footnote{Alternatively, \eqref{arrate} can be interpreted as the limiting form of a discrete spatial approximation consistent with the discrete trajectory based communication model in \cite{wu2018uav, zeng2019uav}. Specifically, the antenna trajectory is partitioned into sufficiently small subintervals, and the instantaneous achievable rate is assumed constant within each subinterval. As the subinterval length approaches zero, the resulting summation converges to \eqref{arrate}.}. It can be observed that the average rate under antenna roaming does not explicitly depend on $T$ or $v_\mathrm{pin}$, as the temporal normalization factor $1/(T v_\mathrm{pin})$ reduces to the spatial normalization factor $1/D_x$ according to \eqref{speed}.

The differences between antenna placement and antenna roaming are summarized as follows.
\begin{remark}
Antenna placement and antenna roaming represent two distinct operation modes for movable pinching antennas. The former is point based, where the optimized service points determine the antenna positions for data transmission, whereas the latter is interval based, where the optimized partition boundaries divide the antenna trajectory into consecutive service intervals.
\end{remark}
\begin{remark}
In pinching-antenna systems with finite antenna speeds, spatial and temporal resource allocations are inherently coupled. Under antenna placement, the service points determine both the channel gains and the transmission times remaining after repositioning. Under antenna roaming, the partition boundaries determine the service intervals, thereby affecting both the channel variation over each interval and the service duration allocated to each user.
\end{remark}

Under antenna roaming, the movement of the pinching antenna may introduce a Doppler-like phase variation. However, the considered antenna movement speed is comparable to typical indoor pedestrian mobility adopted in standardized channel models \cite{3gpp38901}. At such a low movement speed, a $28$~GHz system with $v_\mathrm{pin}=1$~m/s experiences a maximum Doppler shift of approximately $93$~Hz. With a subcarrier spacing of $60$~kHz, this corresponds to a normalized Doppler shift of only $1.6\times10^{-3}$. Moreover, since the antenna trajectory is predetermined, the phase evolution can be tracked and compensated through synchronization and channel estimation. Therefore, the achievable rate is evaluated based on the distance dependent channel power gain $|h_n(x)|^2$, while residual phase-tracking errors are not explicitly modeled.
%%%%%%%%%%%%%%%%%%%%%%%%%%%%%%%%%%%%%%%%%%%%%%%%%
\subsection{Performance Analysis}
To characterize the relative advantages of antenna placement and antenna roaming, their sum average rates are compared in this subsection. For notational convenience, the spatial average of $r_n(x)$ over the service interval of user $n$ is defined as
\begin{equation}
\bar{r}_n^\mathrm{AR}=\frac{1}{x_{n+1}^\mathrm{bd}\!-\!x_n^\mathrm{bd}}\int_{x_n^\mathrm{bd}}^{x_{n+1}^\mathrm{bd}}\!r_n(x)dx,
\end{equation}
where $x_{n+1}^\mathrm{bd}-x_n^\mathrm{bd}>0$. Accordingly, the average rate of user $n$ can be rewritten as follows:
\begin{equation}\label{averate}
R_n^\mathrm{AR}(x_n^\mathrm{bd},x_{n+1}^\mathrm{bd})=\frac{x_{n+1}^\mathrm{bd}\!-\!x_n^\mathrm{bd}}{D_x}\bar{r}_n^\mathrm{AR}.
\end{equation}
The sum-rate gain of antenna roaming over antenna placement is characterized below.
\begin{proposition}\label{gain}
For any feasible antenna placement solution $\{x_n^\mathrm{pt}\}_{n=1}^{N}$ and antenna roaming solution $\{x_n^\mathrm{bd}\}_{n=1}^{N+1}$, the sum-rate gain of antenna roaming over antenna placement, denoted by $\Delta R_\mathrm{sum}$, can be expressed as follows:
\begin{align}
\Delta R_\mathrm{sum}&=\sum_{n=1}^{N}\!\left(\frac{x_{n+1}^\mathrm{bd}\!-\!x_n^\mathrm{bd}}{D_x}\!-\!\frac{1}{N}\!+\!\frac{x_n^\mathrm{pt}\!-\!x_{n-1}^\mathrm{pt}}{T v_\mathrm{max}}\right)\bar{r}_n^\mathrm{AR}\nonumber\\
&\quad-\!\sum_{n=1}^{N}\!\left(\frac{1}{N}\!-\!\frac{x_n^\mathrm{pt}\!-\!x_{n-1}^\mathrm{pt}}{T v_\mathrm{max}}\right)\!\left(r_n(x_n^\mathrm{pt})\!-\!\bar{r}_n^\mathrm{AR}\right).
\end{align}
\end{proposition}
\begin{IEEEproof}
Based on the average rate expression in \eqref{averate}, the sum average rate under antenna roaming is given by
\begin{equation}
R_\mathrm{sum}^\mathrm{AR}=\sum_{n=1}^{N}\frac{x_{n+1}^\mathrm{bd}-x_n^\mathrm{bd}}{D_x}\bar{r}_n^\mathrm{AR}.
\end{equation}
The sum average rate under antenna placement is
\begin{equation}
R_\mathrm{sum}^\mathrm{AP}=\sum_{n=1}^{N}\left(\frac{1}{N}-\frac{x_n^\mathrm{pt}-x_{n-1}^\mathrm{pt}}{T v_\mathrm{max}}\right)r_n(x_n^\mathrm{pt}).
\end{equation}
Therefore, the sum-rate gain can be obtained as follows:
\begin{equation}
\Delta R_\mathrm{sum}\!=\!\sum_{n=1}^{N}\!\frac{x_{n+1}^\mathrm{bd}\!\!-\!x_n^\mathrm{bd}}{D_x}\bar{r}_n^\mathrm{AR}\!-\!\!\sum_{n=1}^{N}\!\!\left(\!\frac{1}{N}\!-\!\frac{x_n^\mathrm{pt}\!-\!x_{n-1}^\mathrm{pt}}{T v_\mathrm{max}}\!\right)\!r_n(x_n^\mathrm{pt}).
\end{equation}
By adding and subtracting $\sum_{n=1}^{N}\!\!\left(\frac{1}{N}\!-\!\frac{x_n^\mathrm{pt}-x_{n-1}^\mathrm{pt}}{T v_\mathrm{max}}\right)\bar{r}_n^\mathrm{AR}$, the desired expression is obtained, which completes the proof.
\end{IEEEproof}
Note that although $\Delta R_\mathrm{sum}$ is referred to as a sum-rate gain, it is not guaranteed to be positive for all feasible antenna placement and antenna roaming solutions. The decomposition in Proposition~\ref{gain} is interpreted in the following remark.
\begin{remark}
In Proposition~\ref{gain}, the term
\begin{equation}
\sum_{n=1}^{N}\Bigg[\underbrace{\frac{x_{n+1}^\mathrm{bd}\!-\!x_n^\mathrm{bd}}{D_x}}_{(a)}-\Bigg(\underbrace{\frac{1}{N}\!-\!\frac{x_n^\mathrm{pt}\!-\!x_{n-1}^\mathrm{pt}}{T v_\mathrm{max}}}_{(b)}\Bigg)\Bigg]\bar{r}_n^\mathrm{AR}
\end{equation}
characterizes the rate difference due to effective transmission time, determined by (a) the service-time fraction under antenna roaming and (b) the normalized effective transmission time under antenna placement. In contrast, the term
\begin{equation}
\sum_{n=1}^{N}\Bigg(\underbrace{\frac{1}{N}\!-\!\frac{x_n^\mathrm{pt}\!-\!x_{n-1}^\mathrm{pt}}{T v_\mathrm{max}}}_{(c)}\Bigg)\big(\underbrace{\vphantom{\frac{x_n^\mathrm{pt}\!-\!x_{n-1}^\mathrm{pt}}{T v_\mathrm{max}}}r_n(x_n^\mathrm{pt})\!-\!\bar{r}_n^\mathrm{AR}}_{(d)}\big)
\end{equation}
represents the channel induced rate difference between point based and interval based transmission, determined by (c) the normalized effective transmission time under antenna placement and (d) the corresponding achievable-rate difference.
\end{remark}

To further interpret the two terms in Proposition~\ref{gain}, a local approximation is derived for service points and service intervals sufficiently close to the corresponding user projections.
\begin{proposition}\label{approxgain}
By applying a second-order Taylor approximation of $r_n(x)$ around the user projection $x_n$, the sum-rate gain of antenna roaming over antenna placement is approximated as follows:
\begin{align}
\Delta R_\mathrm{sum}\!\approx &\sum_{n=1}^{N}\!\left(\frac{x_{n+1}^\mathrm{bd}\!-\!x_n^\mathrm{bd}}{D_x}\!-\!\frac{1}{N}\!+\!\frac{x_n^\mathrm{pt}\!-\!x_{n-1}^\mathrm{pt}}{Tv_\mathrm{max}}\right)r_n(x_n)\nonumber\\
&+\!\sum_{n=1}^{N}\!\left(\frac{1}{N}\!-\!\frac{x_n^\mathrm{pt}\!-\!x_{n-1}^\mathrm{pt}}{Tv_\mathrm{max}}\right)\kappa_n\!\left(x_n^\mathrm{pt}\!-\!x_n\right)^2\\
&-\!\sum_{n=1}^{N}\!\frac{x_{n+1}^\mathrm{bd}\!\!-\!x_n^\mathrm{bd}}{D_x}\kappa_n\!\!\left[\left(x_n^\mathrm{cen}\!-\!x_n\right)^2\!\!+\!\frac{\left(x_{n+1}^\mathrm{bd}\!\!-\!x_n^\mathrm{bd}\right)^2}{12}\right],\nonumber
\end{align}
where $\kappa_n=\rho\eta^2/[\ln 2(y_n^2+d^2)(y_n^2+d^2+\rho\eta^2)]$ and $x_n^\mathrm{cen}=(x_n^\mathrm{bd}+x_{n+1}^\mathrm{bd})/2$.
\end{proposition}
\begin{IEEEproof}
Refer to Appendix~A.
\end{IEEEproof}

Proposition~\ref{approxgain} provides a local geometric interpretation of the channel related terms in the sum-rate gain, which leads to the following remark.
\begin{remark}
Compared with Proposition~\ref{gain}, Proposition~\ref{approxgain} further characterizes the channel related terms through the squared position deviations from the user projections. The term
\begin{equation}
\sum_{n=1}^{N}\!\left(\frac{1}{N}\!-\!\frac{x_n^\mathrm{pt}\!-\!x_{n-1}^\mathrm{pt}}{Tv_\mathrm{max}}\right)\kappa_n\!\left(x_n^\mathrm{pt}\!-\!x_n\right)^2
\end{equation}
represents the rate loss of antenna placement caused by the deviation of the service point $x_n^\mathrm{pt}$ from the user projection $x_n$. In contrast, the term
\begin{equation}
\sum_{n=1}^{N}\!\frac{x_{n+1}^\mathrm{bd}\!-\!x_n^\mathrm{bd}}{D_x}\kappa_n\!\!\left[\left(x_n^\mathrm{cen}\!-\!x_n\right)^2\!\!+\!\frac{\left(x_{n+1}^\mathrm{bd}\!-\!x_n^\mathrm{bd}\right)^2}{12}\right]
\end{equation}
corresponds to the channel loss incurred by antenna roaming, which increases with the interval-center offset $|x_n^\mathrm{cen}-x_n|$ and the service-interval length $x_{n+1}^\mathrm{bd}-x_n^\mathrm{bd}$.
\end{remark}
%%%%%%%%%%%%%%%%%%%%%%%%%%%%%%%%%%%%%%%%%%%%%%%%%
%%%%%%%%%%%%%%%%%%%%%%%%%%%%%%%%%%%%%%%%%%%%%%%%%
\section{Problem Formulation}
In this section, the sum rate maximization problems for antenna placement and antenna roaming are formulated. The objective is to maximize the sum average rate of all users over one scheduling cycle while satisfying the individual minimum rate requirements. To this end, QoS constraints are imposed by specifying a common target rate $R_\mathrm{min}$ for all users.
%%%%%%%%%%%%%%%%%%%%%%%%%%%%%%%%%%%%%%%%%%%%%%%%%
\subsection{Service-Point Deployment for Antenna Placement}
For antenna placement, the service-point deployment problem optimizes the user-specific service points of the pinching antenna. Since the pinching antenna moves from left to right during the considered scheduling cycle, the service points must preserve the same ordering. Moreover, each service point must be reachable within the time slot allocated to the corresponding user. Accordingly, the sum rate maximization problem for antenna placement is formulated as follows:
\begin{subequations}
\begin{empheq}{align}
\max_{\{x_n^\mathrm{pt}\}_{n=1}^N}\quad & \sum_{n=1}^N R_n^\mathrm{AP}(x_{n-1}^\mathrm{pt},x_n^\mathrm{pt})\\
\textrm{s.t.}\quad 
& -\frac{D_x}{2}\leq x_n^\mathrm{pt}\leq \frac{D_x}{2},\,\forall n \in\mathcal{N},\\
& x_{n-1}^\mathrm{pt}\le x_n^\mathrm{pt},\,\forall n \in\mathcal{N},\\
& x_n^\mathrm{pt}-x_{n-1}^\mathrm{pt} \le v_\mathrm{max}\frac{T}{N},\,\forall n \in\mathcal{N},\\
& R_n^\mathrm{AP}(x_{n-1}^\mathrm{pt},x_n^\mathrm{pt}) \ge R_\mathrm{min},\,\forall n \in\mathcal{N}.
\end{empheq}
\label{approblem}
\end{subequations}\vspace{-2mm}\\
In problem \eqref{approblem}, constraint (\ref{approblem}b) confines each service point to the waveguide region. Constraints (\ref{approblem}c) and (\ref{approblem}d) jointly characterize left-to-right repositioning at a finite speed, ensuring that the pinching antenna can reach each service point within the corresponding time slot. Constraint (\ref{approblem}e) ensures that the average rate of each user meets the target rate.
%%%%%%%%%%%%%%%%%%%%%%%%%%%%%%%%%%%%%%%%%%%%%%%%%
\subsection{Service-Interval Partitioning for Antenna Roaming}
For antenna roaming, the service-interval partitioning problem optimizes the internal partition boundaries. These boundaries divide the antenna trajectory into user-specific service intervals, while the two endpoints are fixed at $x_1^\mathrm{bd}=-D_x/2$ and $x_{N+1}^\mathrm{bd}=D_x/2$. Accordingly, the sum rate maximization problem for antenna roaming is formulated as follows:
\begin{subequations}
\begin{empheq}{align}
\max_{\{x_n^\mathrm{bd}\}_{n=2}^N}\quad & \sum_{n=1}^N R_n^\mathrm{AR}(x_n^\mathrm{bd},x_{n+1}^\mathrm{bd})\\
\textrm{s.t.}\quad 
& x_n^\mathrm{bd}\le x_{n+1}^\mathrm{bd},\,\forall n \in\mathcal{N},\\
& R_n^\mathrm{AR}(x_n^\mathrm{bd},x_{n+1}^\mathrm{bd}) \ge R_\mathrm{min},\,\forall n \in\mathcal{N}.
\end{empheq}
\label{arproblem}
\end{subequations}\vspace{-2mm}\\
In problem \eqref{arproblem}, constraint (\ref{arproblem}b) preserves the left-to-right ordering of the partition boundaries while restricting them to the waveguide region. Constraint (\ref{arproblem}c) imposes the target-rate requirement on every user. In addition, the antenna speed should satisfy $v_\mathrm{pin}\le v_\mathrm{max}$. Since $v_\mathrm{pin}=D_x/T$ is fixed under the considered constant-speed antenna roaming model, this requirement is treated as a system-level feasibility condition and is therefore not explicitly included in problem \eqref{arproblem}.

Although both problems involve optimizing spatial points along the waveguide, these points have distinct physical interpretations. In problem \eqref{approblem}, the service points determine the tradeoff between repositioning delay and effective data transmission time under a finite antenna speed. By contrast, in problem \eqref{arproblem}, the partition boundaries define the user-specific service intervals, thereby affecting both the service duration and the channel variation experienced by each user.
%%%%%%%%%%%%%%%%%%%%%%%%%%%%%%%%%%%%%%%%%%%%%%%%%
%%%%%%%%%%%%%%%%%%%%%%%%%%%%%%%%%%%%%%%%%%%%%%%%%
\section{Solution for Service-Point Deployment}
Under antenna placement, the objective function of problem \eqref{approblem} shows that each user's contribution to the sum average rate depends only on two consecutive service points. Therefore, problem \eqref{approblem} exhibits an additive chain structure along the movement direction of the pinching antenna. By exploiting this structure, a DP based algorithm \cite{bellman1966dynamic, bertsekas2012dynamic} with discretized candidate service points is developed in this section.

To transform the original continuous-state service-point deployment problem into a finite-state sequential decision problem, the waveguide is uniformly discretized into $K$ candidate points \cite{chen2025pin, sun2026pin}, denoted by $\mathcal{Q}=\{q_1,q_2,\dots,q_K\}$. Specifically, the $k$-th candidate point is given by
\begin{equation}\label{candidateap}
q_k=-\frac{D_x}{2}+\frac{(k-1)D_x}{K-1}.
\end{equation}
In this case, the candidate service points are ordered from $-D_x/2$ to $D_x/2$, and each continuous service point $x_n^\mathrm{pt}$ is restricted to be selected from $\mathcal{Q}$, which automatically satisfies constraint (\ref{approblem}b). Consider a candidate state transition from $q_i$ to $q_j$, where $q_i$ and $q_j$ are selected as the service points of users $n-1$ and $n$, respectively. Based on constraints (\ref{approblem}c)-(\ref{approblem}e), the feasible transition set for user $n$ is defined as follows:
\begin{equation}
\mathcal{A}_n=\left\{\!(q_i, q_j)\!\in\!\mathcal{Q}\!\times\!\mathcal{Q}\,\Bigg|
\begin{array}{l}
\!0\le q_j\!-\!q_i\le v_\mathrm{max}\frac{T}{N},\\
\!R_n^\mathrm{AP}(q_i,q_j)\ge R_\mathrm{min}
\end{array}
\!\!\!\!\right\}.
\end{equation}
The first condition incorporates the left-to-right ordering constraint in (\ref{approblem}c) and the finite-speed constraint in (\ref{approblem}d), while the second imposes the QoS requirement in (\ref{approblem}e). Hence, a transition from $q_i$ to $q_j$ is feasible if and only if $(q_i,q_j)\in\mathcal{A}_n$. Based on this feasible transition set, the reward associated with the transition from $q_i$ to $q_j$ for user $n$ is defined as
\begin{equation}\label{ffunctionap}
F_n(q_i,q_j)\!=\!\begin{cases}\displaystyle
\!R_n^\mathrm{AP}(q_i,q_j),\!\!&\text{if } (q_i,q_j)\!\in\!\mathcal{A}_n,\\
-\infty, &\text{otherwise}.
\end{cases}
\end{equation}
For a feasible transition, $F_n(q_i,q_j)$ is the average rate achieved by user $n$, whereas the value $-\infty$ excludes infeasible transitions from the DP recursion.

Based on the sequential structure, an accumulated reward function, denoted by $V_n(q_j)$, is introduced to characterize the maximum accumulated sum of average rates for the first $n$ users, given that $q_j$ is selected as the service point of user $n$. Following the DP principle, for $n=2,\dots,N$, $V_n(q_j)$ is obtained by selecting the previous candidate service point that maximizes the accumulated reward, as follows:
\begin{equation}\label{accrewardap}
V_n(q_j)=\max_{q_i\in\mathcal{Q}}\left(V_{n-1}(q_i)+F_n(q_i,q_j)\right).
\end{equation}
By assigning $-\infty$ to infeasible transitions and states, infeasible paths are automatically excluded from the DP recursion. For each accumulated reward $V_n(q_j)$, the corresponding service-point path is recorded in $\mathcal{S}_n(q_j)$. Given a feasible current state $q_j$, the predecessor $q_i^\star$ is selected to identify the optimal predecessor path $\mathcal{S}_{n-1}(q_i^\star)$, where
\begin{equation}
q_i^\star\in\arg\max_{q_i\in\mathcal{Q}}\left(V_{n-1}(q_i)+F_n(q_i,q_j)\right).
\end{equation}
The path associated with $V_n(q_j)$ is then updated as
\begin{equation}
\mathcal{S}_n(q_j)=\left[\mathcal{S}_{n-1}(q_i^\star),q_j\right].
\end{equation}
For an infeasible state, $\mathcal{S}_n(q_j)$ is set to $\emptyset$.

For user $1$, with the initial antenna position fixed at $q_1=-D_x/2$, the accumulated reward for each candidate service point $q_j\in\mathcal{Q}$ is initialized as follows:
\begin{equation}
V_1(q_j)=F_1(q_1,q_j),
\end{equation}
where infeasible transitions are assigned $-\infty$ according to \eqref{ffunctionap}. After all $N$ users have been processed, the optimal terminal state is selected as
\begin{equation}
q_\mathrm{fin}^\star\in\arg\max_{q_j\in\mathcal{Q}}V_N(q_j),
\end{equation}
and the optimized service-point path is given by $\mathcal{S}_N(q_\mathrm{fin}^\star)$.

Based on the above formulation, a DP based service-point deployment algorithm is presented in Algorithm~\ref{apalg}. During initialization, the feasible service points for user $1$ are constructed from the initial pinching antenna position, providing the initial states for the DP recursion, as shown in lines 2–9. For each user $n\ge 2$ and each current state $q_j$, the accumulated reward and the corresponding service-point path are updated according to \eqref{accrewardap}. After all users have been processed, the optimal terminal state is selected, and the optimized discretized service-point path is obtained as $\mathcal{S}_N(q_\mathrm{fin}^\star)$, from which the optimized service points $\{x_n^{\mathrm{pt},\star}\}_{n=1}^{N}$ are determined.

\begin{algorithm}[t]
\caption{DP based Service-Point Deployment Algorithm}
\label{apalg}
\begin{algorithmic}[1]
\STATE Construct the candidate set $\mathcal{Q}$ according to \eqref{candidateap}.
\FOR{each $q_j\in\mathcal{Q}$}
\STATE Calculate $V_1(q_j)=F_1(q_1,q_j)$.
\IF{$V_1(q_j)=-\infty$}
\STATE Set $\mathcal{S}_1(q_j)=\emptyset$.
\ELSE
\STATE Set $\mathcal{S}_1(q_j)=[q_j]$.
\ENDIF
\ENDFOR
\FOR{$n=2$ to $N$}
\FOR{each $q_j\in\mathcal{Q}$}
\STATE Calculate $V_n(q_j)=\max_{q_i\in\mathcal{Q}}(V_{n-1}(q_i)+F_n(q_i,q_j))$.
\IF{$V_n(q_j)=-\infty$}
\STATE Set $\mathcal{S}_n(q_j)=\emptyset$.
\ELSE
\STATE Select $q_i^\star\in\arg\max_{q_i\in\mathcal{Q}}(V_{n-1}(q_i)+F_n(q_i,q_j))$.
\STATE Set $\mathcal{S}_n(q_j)=[\mathcal{S}_{n-1}(q_i^\star),q_j]$.
\ENDIF
\ENDFOR
\ENDFOR
\STATE Select $q_\mathrm{fin}^\star=\arg\max_{q_j\in\mathcal{Q}}V_N(q_j)$.
\STATE Output $\mathcal{S}_N(q_\mathrm{fin}^\star)$.
\end{algorithmic}
\end{algorithm}

For each user $n$ and current state $q_j$, the DP recursion searches over the feasible predecessor set $\mathcal{P}_n(q_j)$. Since each state has at most $K$ predecessors, the worst-case computational complexity is $\mathcal{O}(NK^2)$. In practice, the movement and QoS constraints reduce the number of feasible predecessors, thereby lowering the computational complexity. Due to the additive chain structure, the DP based algorithm is globally optimal for the discretized version of problem \eqref{approblem}. For the original continuous problem, the resulting solution is generally suboptimal, while increasing the number of candidate points improves the discretization accuracy at the cost of higher complexity.
%%%%%%%%%%%%%%%%%%%%%%%%%%%%%%%%%%%%%%%%%%%%%%%%%
%%%%%%%%%%%%%%%%%%%%%%%%%%%%%%%%%%%%%%%%%%%%%%%%%
\section{Solution for Service-Interval Partitioning}
This section develops a solution to problem \eqref{arproblem} by exploiting the properties of the instantaneous rate functions. The service-interval structure of antenna roaming is characterized, with the intersections of the instantaneous rate functions determining the optimal partition boundaries for the unconstrained problem. Based on this structural result, a DP based algorithm is developed to solve the service-interval partitioning problem.
%%%%%%%%%%%%%%%%%%%%%%%%%%%%%%%%%%%%%%%%%%%%%%%%%
\subsection{Unconstrained Service-Interval Structure}
As shown in \eqref{arrate}, the sum rate is determined by the partition of the antenna trajectory among the users. Without QoS constraints, the optimal partition can be obtained by assigning each antenna position to the user achieving the highest instantaneous rate. To characterize the resulting partition-boundary structure, the following lemma establishes a single-crossing property of the instantaneous rate functions.
\begin{lemma}\label{crossing}
In the considered pinching-antenna system, for any two users $m$ and $n$ with either $x_m\neq x_n$ or $y_m^2\neq y_n^2$, the instantaneous rate functions $r_m(x)$ and $r_n(x)$ intersect at most once along the waveguide under fixed transmit power.
\end{lemma}
\begin{IEEEproof}
By taking the squared magnitude of the channel coefficient, the phase term is eliminated, giving the following channel power gain for user $n$:
\begin{equation}
|h_n(x)|^2=\frac{\eta^2}{(x-x_n)^2+y_n^2+d^2}.
\end{equation}
Under fixed transmit power, the instantaneous rate depends on the antenna position only through the corresponding antenna-user distance. Since $\log_2(1+z)$ is strictly increasing in $z$, $r_m(x)\ge r_n(x)$ is equivalent to
\begin{equation}
(x-x_m)^2+y_m^2+d^2\le (x-x_n)^2+y_n^2+d^2.
\end{equation}
After simplification, this condition becomes
\begin{equation}
2(x_n-x_m)x+x_m^2-x_n^2+y_m^2-y_n^2\le 0.
\end{equation}
If $x_m\neq x_n$, the left-hand side is a non-constant affine function of $x$ and therefore has at most one root. If $x_m=x_n$ and $y_m^2\neq y_n^2$, the left-hand side reduces to a nonzero constant and therefore has no root. Hence, $r_m(x)=r_n(x)$ holds at most once along the waveguide, which completes the proof.
\end{IEEEproof}
Lemma~\ref{crossing} indicates that the rate ordering between any two users changes at most once along the waveguide. With the users indexed by their $x$-coordinates, this single-crossing property provides the basis for characterizing the boundary structure of the unconstrained service-interval partitioning problem. Specifically, each partition boundary between two nonempty adjacent service intervals is located at an intersection of the corresponding instantaneous rate functions, as stated in the following proposition.
\begin{proposition}\label{uncons}
For antenna roaming without QoS constraints, an optimal service partition assigns each antenna position $x$ to a user satisfying
\begin{equation}
n^\star(x)\in\arg\max_{i\in\mathcal{N}}r_i(x).
\end{equation}
Moreover, if a partition boundary $x^{\mathrm{bd},\star}$ separates two nonempty adjacent service intervals assigned to users $m$ and $n$, then
\begin{equation}
r_m(x^{\mathrm{bd},\star})=r_n(x^{\mathrm{bd},\star})=\max_{i\in\mathcal{N}}r_i(x^{\mathrm{bd},\star}).
\end{equation}
\end{proposition}
\begin{IEEEproof}
Without QoS constraints, the objective in \eqref{arproblem} is the integral of the instantaneous rate selected at each antenna position. Since the assignment at one position does not affect that at another, the objective is maximized by assigning each position to a user achieving the highest instantaneous rate, i.e.,
\begin{equation}
n^\star(x)\in\arg\max_{i\in\mathcal{N}}r_i(x).
\end{equation}
As a result, any user that does not achieve the maximum instantaneous rate at any antenna position is assigned an empty service interval. Moreover, consider an interior partition boundary $x^{\mathrm{bd},\star}$ separating two nonempty adjacent service intervals assigned to users $m$ and $n$. Immediately to the left and right of the partition boundary, $m$ and $n$ achieve the maximum instantaneous rate, respectively. By continuity of the instantaneous rate functions, the following equality holds:
\begin{equation}
r_m(x^{\mathrm{bd},\star})=r_n(x^{\mathrm{bd},\star})=\max_{i\in\mathcal{N}} r_i(x^{\mathrm{bd},\star}),
\end{equation}
which completes the proof.
\end{IEEEproof}

Proposition~\ref{uncons} characterizes the optimal service partition for antenna roaming without QoS constraints. It shows that each antenna position is optimally assigned to the user with the highest instantaneous rate, and that the boundary between two nonempty adjacent service intervals is determined by the equality of their instantaneous rates. Together with Lemma~\ref{crossing}, this proposition reveals the spatial service structure when users are indexed by their $x$-coordinates, which is reflected by the spatial ordering of the service points in antenna placement and the service intervals in antenna roaming. Moreover, the optimal service-interval partition is completely determined by the geometry of the instantaneous rate functions.

However, the unconstrained structure does not incorporate the target-rate requirements in problem \eqref{arproblem}. Under these constraints, each user must be allocated a service interval that achieves the required average rate, and the partition boundaries may therefore deviate from the intersections of the instantaneous rate functions. Accordingly, the constrained service-interval partitioning problem is addressed by the DP based method developed in the next subsection.
%%%%%%%%%%%%%%%%%%%%%%%%%%%%%%%%%%%%%%%%%%%%%%%%%
\subsection{DP based Constrained Service-Interval Partitioning}
Under the QoS constraint~(\ref{arproblem}c), each user must be allocated a service interval that provides at least the minimum average rate. Therefore, the partition boundaries must be optimized by jointly considering the sum rate objective and the QoS constraints. In this subsection, a DP based algorithm is developed to solve the constrained service-interval partitioning problem over discretized partition boundaries.

To enable a finite-state DP search for antenna roaming, the partition boundaries are selected from a uniformly sampled candidate set $\mathcal{P}=\{p_1,p_2,\dots,p_K\}$. Specifically, for $k=1,2,\dots,K$, the candidate boundary $p_k$ is given by
\begin{equation}\label{candidatear}
p_k=-\frac{D_x}{2}+\frac{(k-1)D_x}{K-1}.
\end{equation}
The endpoint boundaries are fixed at $x_1^\mathrm{bd}=-D_x/2$ and $x_{N+1}^\mathrm{bd}=D_x/2$, while the internal partition boundaries $x_n^\mathrm{bd}$, $n=2,\dots,N$, are selected from $\mathcal{P}$. Under antenna roaming, each DP transition determines the service interval of one user. Given two candidate boundaries $p_i,p_j\in\mathcal{P}$, the transition from $p_i$ to $p_j$ for user $n$ corresponds to assigning the service interval $[p_i,p_j]$ to that user. Based on the ordered partition constraint in (\ref{arproblem}b) and the QoS constraint in (\ref{arproblem}c), the feasible transition set for user $n$ is defined as
\begin{equation}
\mathcal{B}_n=\left\{\!(p_i, p_j)\!\in\!\mathcal{P}\!\times\!\mathcal{P}\,\Bigg|
\begin{array}{l}
\!p_i\le p_j,\\
\!R_n^\mathrm{AR}(p_i,p_j)\ge R_\mathrm{min}
\end{array}
\!\!\right\}.
\end{equation}
The first condition ensures that the antenna trajectory is partitioned in the prescribed order, while the second guarantees that user $n$ can achieve the target rate over the service interval $[p_i,p_j]$. Accordingly, the stage reward is defined as
\begin{equation}
F_n(p_i,p_j)=\begin{cases}
R_n^\mathrm{AR}(p_i,p_j), & \text{if } (p_i,p_j)\in\mathcal{B}_n,\\
-\infty, & \text{otherwise}.
\end{cases}
\end{equation}
Thus, a feasible transition contributes the average rate of user $n$ over $[p_i,p_j]$ to the accumulated reward, whereas an infeasible transition is assigned $-\infty$. Since the starting boundary $p_i$ of user $n$ is also the ending boundary of user $n-1$, its effect on the preceding service intervals is captured by the accumulated reward associated with state $p_i$ in the previous DP stage rather than by the current stage reward $F_n(p_i,p_j)$.

Let $V_n(p_j)$ denote the maximum accumulated reward for users $1,\dots,n$, given that $p_j$ is the ending boundary of user $n$. For each candidate ending boundary $p_j$, the DP recursion searches over all candidate starting boundaries $p_i$ of user $n$, where $p_i$ also serves as the ending boundary of user $n-1$. For $n=2,\dots,N$, the recursion is given by
\begin{equation}\label{accrewardcs}
V_n(p_j)=\max_{p_i\in\mathcal{P}}\left(V_{n-1}(p_i)+F_n(p_i,p_j)\right).
\end{equation}
For a feasible state $p_j$, the optimal predecessor boundary $p_i^\star$ is selected as follows:
\begin{equation}\label{statecs}
p_i^\star\in\arg\max_{p_i\in\mathcal{P}}\left(V_{n-1}(p_i)+F_n(p_i,p_j)\right).
\end{equation}
Let $\mathcal{S}_n(p_j)$ denote the partition boundary path associated with $V_n(p_j)$. It is updated recursively as
\begin{equation}
\mathcal{S}_n(p_j)=\left[\mathcal{S}_{n-1}(p_i^\star),p_j\right].
\end{equation}
For an infeasible state, $\mathcal{S}_n(p_j)=\emptyset$.

For user $1$, the starting boundary is fixed to the left endpoint of the waveguide, i.e., $x_1^\mathrm{bd}=-D_x/2=p_1$. Therefore, for each candidate ending boundary $p_j\in\mathcal{P}$, the accumulated reward is initialized as follows:
\begin{equation}\label{accrewardcs1}
V_1(p_j)=F_1(p_1,p_j).
\end{equation}
The corresponding partition boundary path is initialized as
\begin{equation}
\mathcal{S}_1(p_j)=[p_1,p_j],
\end{equation}
if $V_1(p_j)>-\infty$; otherwise, $\mathcal{S}_1(p_j)=\emptyset$. For user $N$, the ending boundary is fixed at $x_{N+1}^\mathrm{bd}=D_x/2=p_K$. After the recursion reaches user $N$, the optimized partition boundary sequence is obtained as
\begin{equation}
\mathcal{S}^\star=\mathcal{S}_N(p_K).
\end{equation}
The corresponding optimal objective value of the discretized problem is $V_N(p_K)$.

The overall procedure of the proposed DP based service-interval partitioning algorithm for antenna roaming is summarized in Algorithm~\ref{aralg}. The boundary candidate set is first constructed over the antenna trajectory. With the starting boundary of user $1$ fixed at $-D_x/2$ and the ending boundary of user $N$ fixed at $D_x/2$, the DP recursion sequentially updates the accumulated reward and the associated partition boundary path for each user and candidate ending boundary. Finally, the optimized discretized partition boundary sequence is obtained as $\mathcal{S}_N(p_K)$. For each user $n$ and candidate ending boundary $p_j$, the recursion searches over all candidate starting boundaries $p_i\in\mathcal{P}$, resulting in a worst-case computational complexity of $\mathcal{O}(NK^2)$, excluding the numerical evaluation of the average rate terms.

\begin{algorithm}[t]
\caption{DP based Service-Interval Partitioning Algorithm}
\label{aralg}
\begin{algorithmic}[1]
\STATE Construct the boundary candidate set $\mathcal{P}$ according to \eqref{candidatear}.
\FOR{each $p_j\in\mathcal{P}$}
\STATE Calculate $V_1(p_j)=F_1(p_1,p_j)$.
\IF{$V_1(p_j)=-\infty$}
\STATE Set $\mathcal{S}_1(p_j)=\emptyset$.
\ELSE
\STATE Set $\mathcal{S}_1(p_j)=[p_1,p_j]$.
\ENDIF
\ENDFOR
\FOR{$n=2$ to $N$}
\FOR{each $p_j\in\mathcal{P}$}
\STATE Calculate $V_n(p_j)\!=\!\max_{p_i\in\mathcal{P}}(V_{n-1}(p_i)+F_n(p_i,p_j))$.
\IF{$V_n(p_j)=-\infty$}
\STATE Set $\mathcal{S}_n(p_j)=\emptyset$.
\ELSE
\STATE Select $p_i^\star\in\arg\max_{p_i\in\mathcal{P}}(V_{n-1}(p_i)+F_n(p_i,p_j))$.
\STATE Set $\mathcal{S}_n(p_j)=[\mathcal{S}_{n-1}(p_i^\star),p_j]$.
\ENDIF
\ENDFOR
\ENDFOR
\STATE Output $\mathcal{S}_N(p_K)$.
\end{algorithmic}
\end{algorithm}
%%%%%%%%%%%%%%%%%%%%%%%%%%%%%%%%%%%%%%%%%%%%%%%%%
%%%%%%%%%%%%%%%%%%%%%%%%%%%%%%%%%%%%%%%%%%%%%%%%%
\section{Simulation Results}
In this section, simulation results are presented to evaluate the proposed antenna placement and antenna roaming schemes and the gains achieved by the corresponding DP based algorithms. In the simulations, the users are randomly distributed on the ground within the service region. Three baseline schemes are considered for comparison. The \emph{fixed antenna} scheme employs a single antenna fixed at the origin with height $d$. The \emph{ideal} scheme adopts conventional antenna placement with infinite antenna movement speed. The \emph{uniform} scheme follows a fixed trajectory, with the antenna sequentially moving to $N$ uniformly spaced service points for antenna placement and traversing $N$ equal-length service intervals for antenna roaming. When the DP based solution is infeasible, the corresponding uniform solution is adopted instead. The main simulation parameters are summarized in Table~\ref{parameter}.

\begin{table}[t]
\centering
\caption{Simulation Parameters}\vspace{-1mm}
\label{parameter}
\begin{tabular}{lc}
\hline
\textbf{Parameter} & \textbf{Value} \\ \hline
Carrier frequency ($f_c$) & $28$~GHz\\
Noise power ($\sigma^2$) & $-90$~dBm\\ 
Effective refractive index ($n_\mathrm{eff}$) & $1.4$ \\
Waveguide height ($d$) & $3$~m \\
Service area ($D_x \times D_y$) & $20$~m $\times$ $10$~m \\
Maximum antenna movement speed ($v_\mathrm{max}$) & $5$~m/s \\ 
Number of candidate points ($K$) & 100 \\ \hline
\end{tabular}\vspace{-4mm}
\end{table}

\begin{figure}[!t]
\centering
\subfigure[User Distribution]{\includegraphics[width=85mm, trim=0 3mm 0 0, clip]{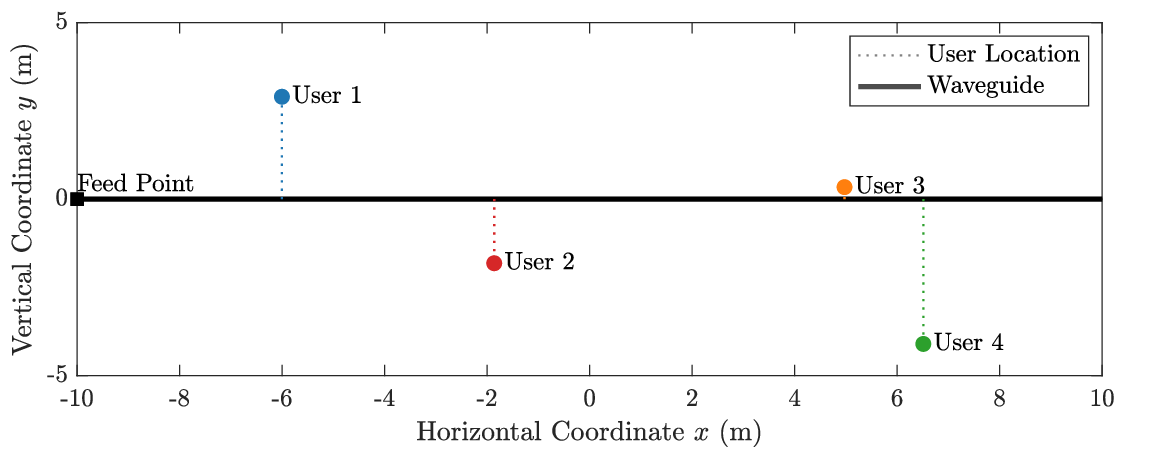}}\vspace{-2mm}
\subfigure[DP based Solution]{\hspace{-4mm}\includegraphics[width=95mm, trim=0 3mm 0 0, clip]{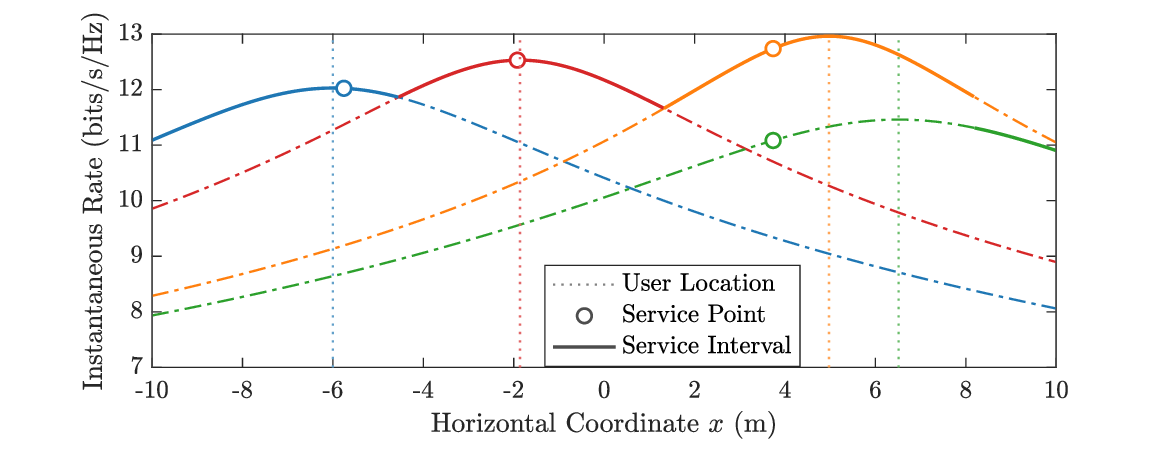}}\vspace{-2mm}
\caption{Illustration of the optimized antenna placement and antenna roaming for a given user distribution, where $N=4$, $T=20$~s, $P_t=20$~dBm, and $R_\mathrm{min}=1$~bits/s/Hz.}\vspace{-4mm}
\label{result0}
\end{figure}

\fref{result0} illustrates the DP based antenna placement and antenna roaming solutions for a representative user distribution. Under antenna placement, the optimized service points deviate from the rate-maximizing positions, i.e., the users' projections onto the waveguide, because the service-point design balances antenna repositioning time and data transmission time. In particular, the service point of user $1$ is shifted rightward from its waveguide projection to shorten the subsequent repositioning distance toward user $2$, who is closer to the waveguide and benefits more from antenna repositioning. For user $3$, the rate-maximizing position is far from the preceding service point, and an intermediate location is therefore selected to balance the achievable rate gain and repositioning delay. For user $4$, further movement provides only marginal gain due to its large distance from the waveguide, and thus transmission begins without additional repositioning. On the other hand, under antenna roaming, the optimized service intervals are consistent with Lemma~\ref{crossing} and Proposition~\ref{uncons}. Specifically, each antenna position tends to be assigned to the user achieving the highest instantaneous rate, with the boundaries between adjacent service intervals lying at the intersections of the corresponding rate curves. For user $4$, since its instantaneous rate is lower than that of user $3$ over the entire waveguide, a minimum service interval is allocated to satisfy the target-rate constraint. Furthermore, the figure supports the performance analysis by illustrating the tradeoff between channel quality and effective transmission time. That is, antenna placement benefits from transmission at high-rate service points, whereas antenna roaming benefits from uninterrupted transmission over service intervals despite a lower spatially averaged rate.

\begin{figure}[!t]
\centering
\subfigure[Sum Rate]{\hspace{-3mm}\includegraphics[width=95mm]{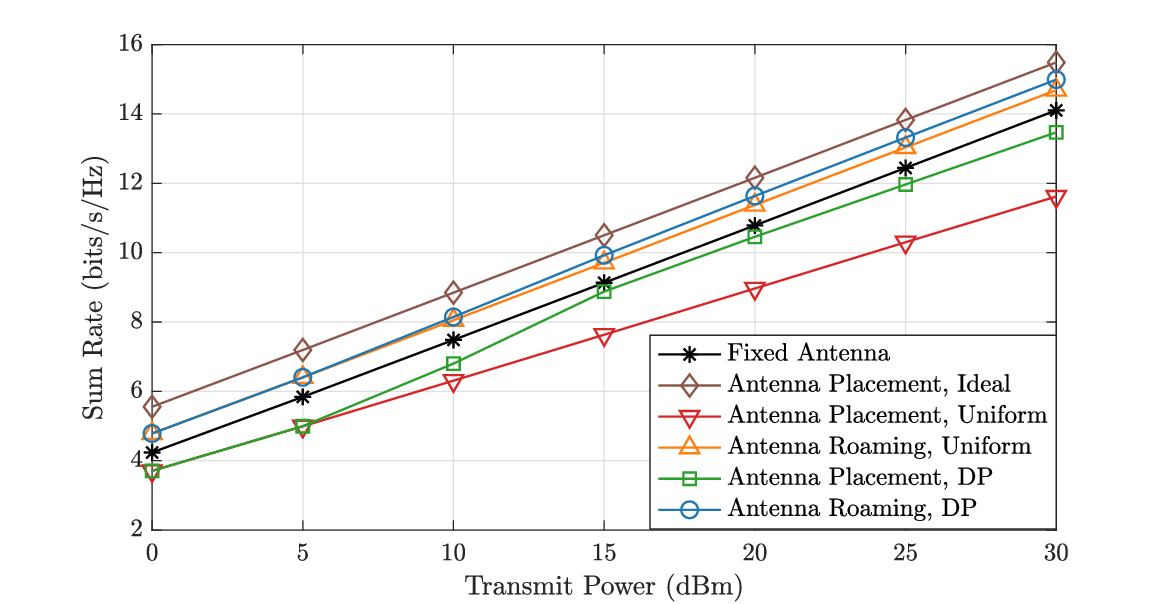}}\vspace{-2mm}
\subfigure[Outage Probability]{\hspace{-3mm}\includegraphics[width=95mm]{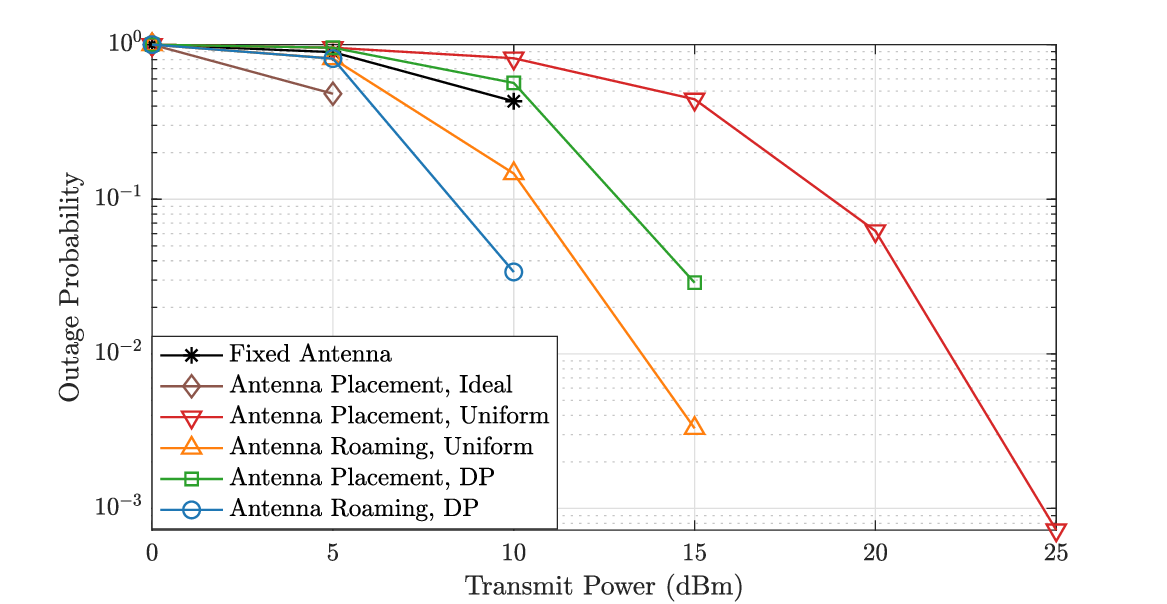}}\vspace{-2mm}
\caption{Impact of the transmit power on the sum rate and outage probability, where $N=4$, $T=20$~s, and $R_\mathrm{min}=1.8$~bits/s/Hz.}\vspace{-4mm}
\label{result1}
\end{figure}

\fref{result1} compares the performance of different schemes as the transmit power increases. As shown in \fref{result1}(a), the sum rates of all schemes increase with transmit power. The ideal antenna placement achieves the highest sum rate by eliminating repositioning overhead, while the DP based schemes consistently outperform their uniform counterparts. The improvement is more pronounced for antenna placement, since service-point optimization can jointly reduce repositioning delay and improve channel quality, whereas antenna roaming already benefits from continuous transmission under uniform partitioning. Meanwhile, DP based antenna roaming approaches the performance of ideal antenna placement, indicating that continuous transmission can effectively compensate for the reduction in instantaneous rate. Under the considered settings, DP based antenna placement performs worse than the fixed-antenna scheme because antenna repositioning consumes part of the scheduling cycle. As shown in \fref{result1}(b), the outage probability decreases with transmit power for all schemes. The DP based schemes achieve lower outage probabilities than the corresponding uniform schemes, confirming the effectiveness of optimizing the spatial service variables under the target-rate constraints. Among the practical schemes with finite antenna movement speed, DP based antenna roaming achieves the lowest outage probability because data transmission continues throughout the antenna movement process.

\begin{figure}[!t]
\hspace{-3mm}\includegraphics[width=95mm]{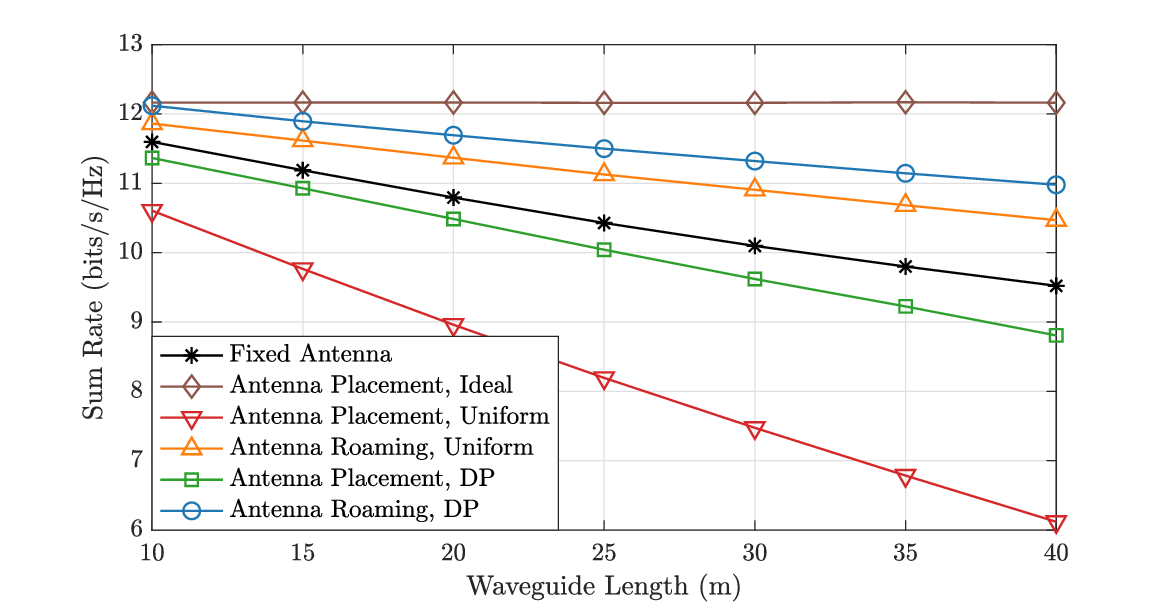}\vspace{-2mm}
\caption{Impact of the waveguide length on the sum rate, where $N=4$, $T=20$~s, $P_t=20$~dBm, and $R_\mathrm{min}=1$~bits/s/Hz.}\vspace{-4mm}
\label{result2}
\end{figure}

\fref{result2} demonstrates the impact of the waveguide length on the sum rate, with the service region extended accordingly. As the waveguide becomes longer, the sum rates of all practical schemes decrease, while ideal antenna placement remains nearly unchanged because repositioning overhead is not considered. For antenna roaming, the performance degradation is mainly attributed to the increased movement range, which exposes the antenna to more positions with relatively low instantaneous rates. By contrast, antenna placement suffers from both the reduced channel efficiency over a larger spatial range and the increased repositioning time. Consequently, antenna roaming becomes increasingly advantageous as the waveguide length grows.

\begin{figure}[!t]
\hspace{-3mm}\includegraphics[width=95mm]{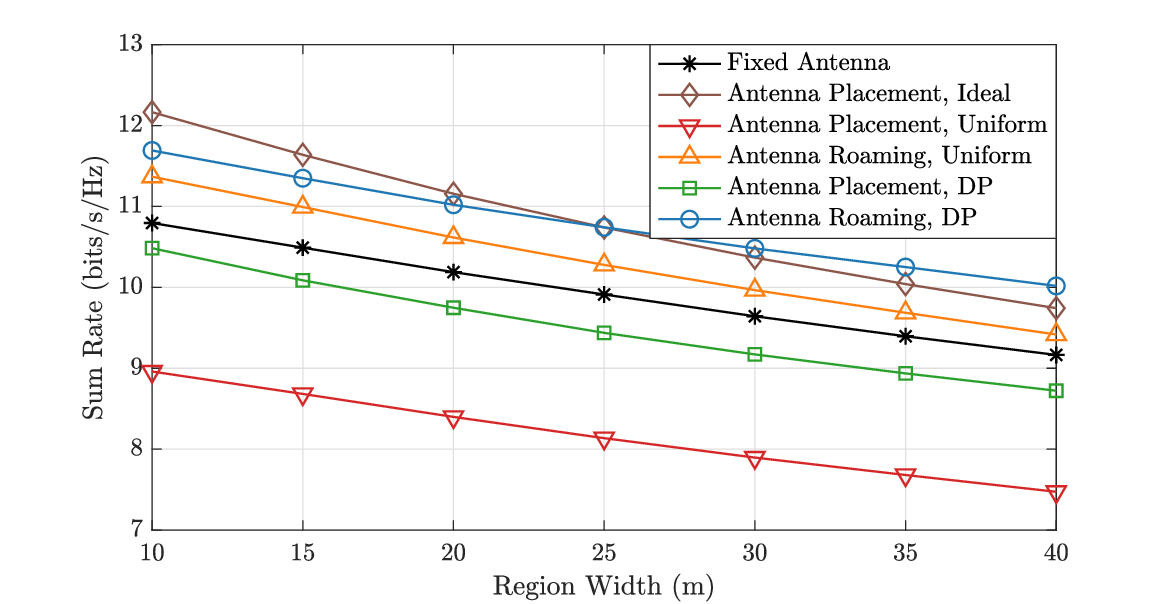}\vspace{-2mm}
\caption{Impact of the region width on the sum rate, where $N=4$, $T=20$~s, $P_t=20$~dBm, and $R_\mathrm{min}=1$~bits/s/Hz.}\vspace{-4mm}
\label{result3}
\end{figure}

In \fref{result3}, the performance of different schemes is evaluated with respect to the service-region width. As the width increases, the users are located farther from the waveguide on average, leading to weaker channels and lower sum rates for all schemes. Moreover, the instantaneous-rate curves of distant users become flatter along the waveguide because antenna movement induces less variation in the propagation distance and hence in the channel gain, as illustrated by user $4$ in \fref{result0}. In this case, antenna roaming can optimize the service-interval lengths to allocate more transmission time to users with higher rates, while assigning weaker users only the intervals required to satisfy their target-rate constraints. As a result, when the service-region width exceeds approximately $25$~m, antenna roaming can outperform ideal antenna placement, which retains equal service durations.

\begin{figure}[!t]
\hspace{-3mm}\includegraphics[width=95mm]{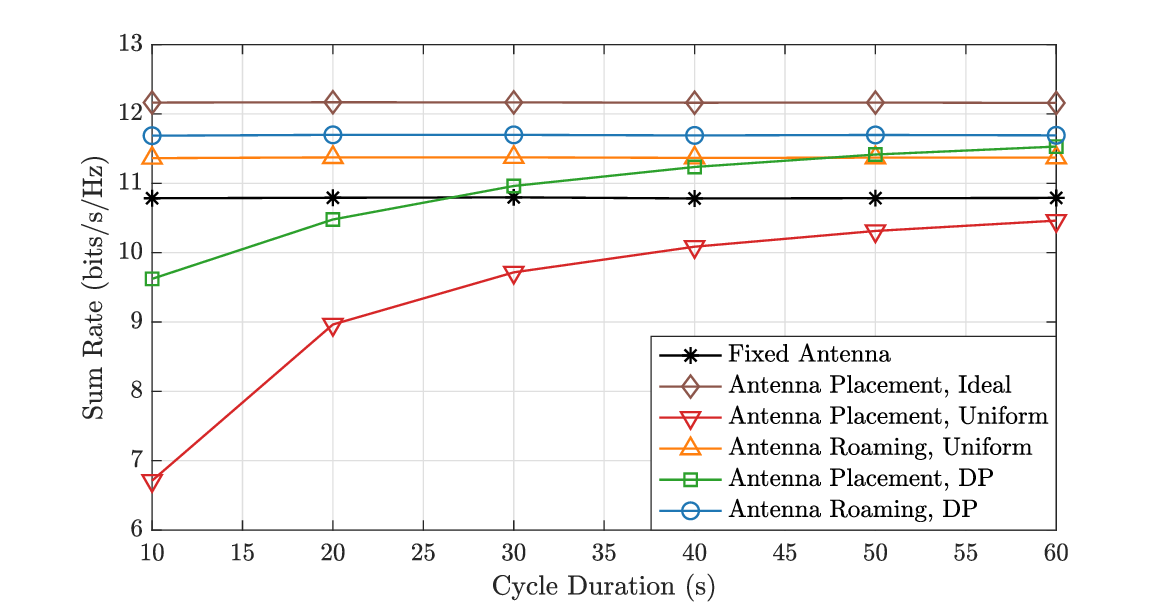}\vspace{-2mm}
\caption{Impact of the cycle duration on the sum rate, where $N=4$, $P_t=20$~dBm, and $R_\mathrm{min}=1$~bits/s/Hz.}\vspace{-4mm}
\label{result4}
\end{figure}

\fref{result4} shows the influence of the cycle duration on different schemes. Except for antenna placement with finite antenna movement speed, the performance of the other schemes is independent of the cycle duration. For antenna placement, increasing the cycle duration reduces the fraction of time consumed by antenna repositioning and correspondingly increases the effective transmission time. Hence, its achievable sum rate increases monotonically with the cycle duration. Under uniform spatial allocation, antenna roaming still achieves a higher sum rate than antenna placement, indicating that its performance advantage remains even without optimizing user service durations through the service-interval lengths.

\begin{figure}[!t]
\hspace{-3mm}\includegraphics[width=95mm]{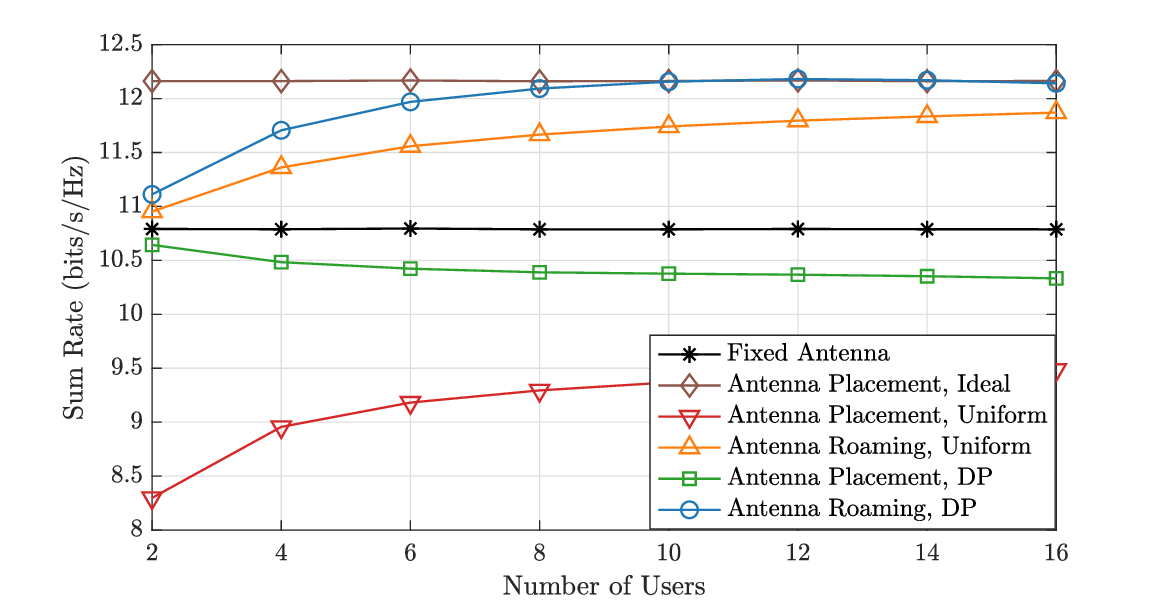}\vspace{-2mm}
\caption{Impact of the number of users on the sum rate, where $T=20$~s, $P_t=20$~dBm, and $R_\mathrm{min}=0.5$~bits/s/Hz.}\vspace{-4mm}
\label{result5}
\end{figure}

The impact of the number of users is presented in \fref{result5}. As the number of users increases, uniform antenna placement and both antenna roaming schemes initially benefit from finer spatial service allocation, which enables more effective exploitation of spatial channel variations. In particular, DP based antenna roaming approaches ideal antenna placement and then decreases slightly when the number of users becomes large, as the target-rate requirements restrict the flexibility of service-interval allocation. By contrast, DP based antenna placement gradually decreases because more frequent repositioning reduces the effective transmission time, while the fixed-antenna and ideal-placement schemes remain nearly unchanged.

\begin{figure}[!t]
\centering
\subfigure[Sum Rate]{\hspace{-3mm}\includegraphics[width=95mm]{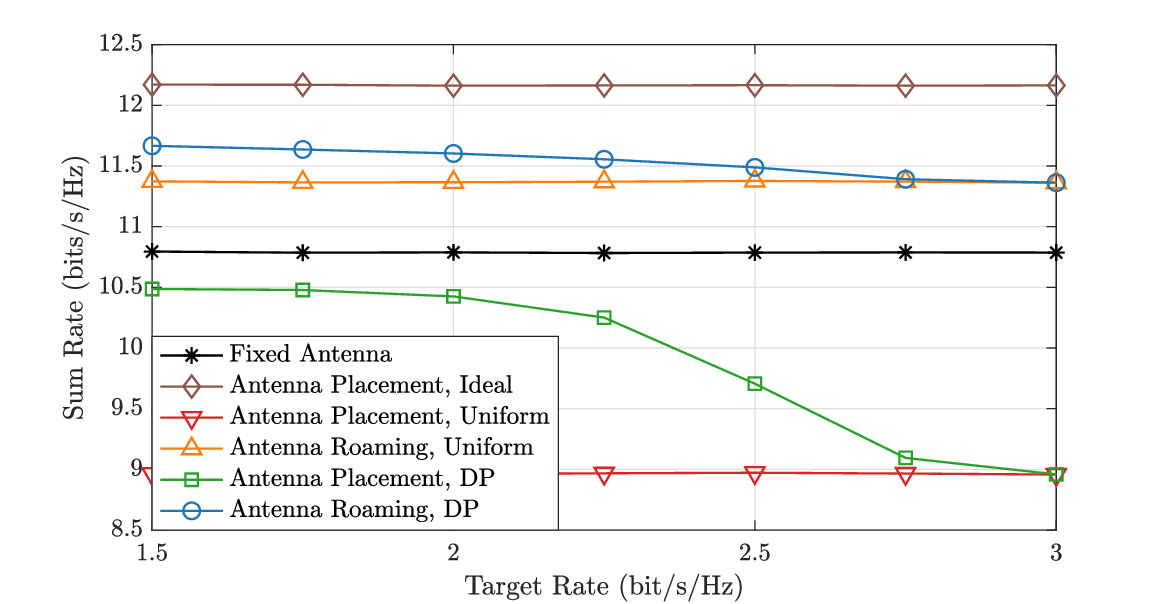}}\vspace{-2mm}
\subfigure[Outage Probability]{\hspace{-3mm}\includegraphics[width=95mm]{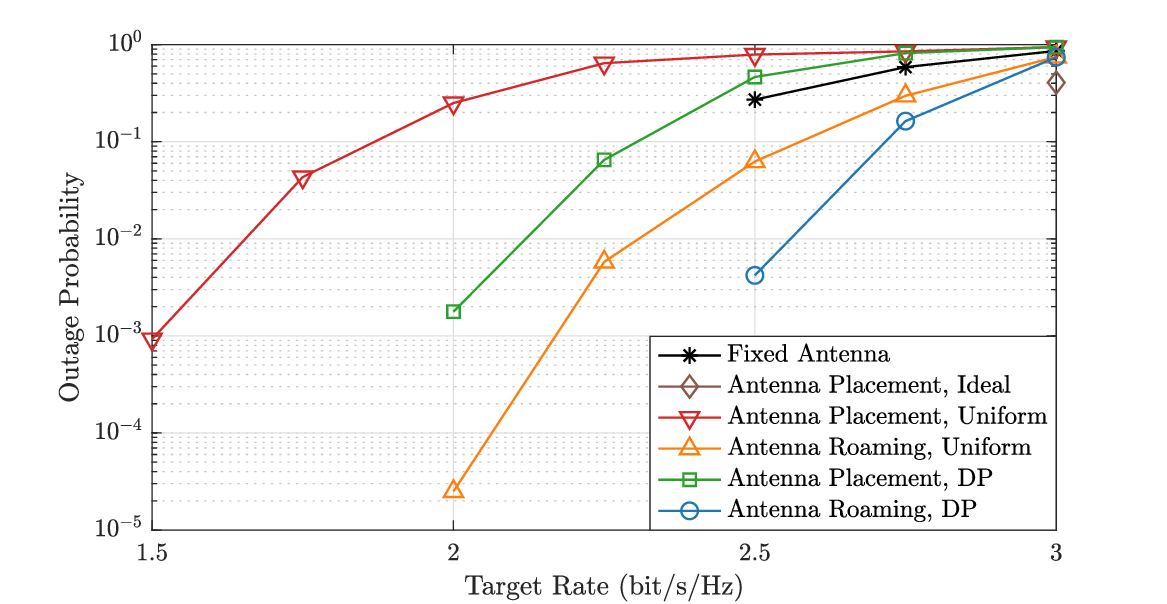}}\vspace{-2mm}
\caption{Impact of the target rate on the sum rate and outage probability, where $N=4$, $T=20$~s, and $P_t=20$~dBm.}\vspace{-4mm}
\label{result6}
\end{figure}

\fref{result6} demonstrates the tradeoff between QoS provisioning and communication performance under different target-rate requirements. As the target rate increases, the proposed DP based algorithms adapt the service points and intervals to satisfy the more stringent rate requirements. This adaptation improves QoS feasibility but reduces the flexibility for sum rate maximization, leading to a gradual decrease in the achievable sum rate. When the target rate becomes excessively high, the discretized DP problem may become infeasible, in which case the corresponding uniform solution is adopted as specified in the simulation setup.
%%%%%%%%%%%%%%%%%%%%%%%%%%%%%%%%%%%%%%%%%%%%%%%%%
%%%%%%%%%%%%%%%%%%%%%%%%%%%%%%%%%%%%%%%%%%%%%%%%%
\section{Conclusions}
This paper investigated pinching-antenna systems with finite antenna movement speed and proposed antenna roaming as an alternative operation mode to antenna placement. A unified cycle-duration framework was established to characterize the tradeoff between effective transmission time and channel quality. Corresponding sum rate maximization problems were formulated and solved using DP. Simulation results validated the theoretical analysis and demonstrated that antenna roaming consistently outperforms antenna placement, with more pronounced gains as positioning overhead increases or the channel-quality improvement achieved through repositioning decreases. These results indicate the potential of antenna roaming as a practical operation mode for future pinching-antenna systems with finite movement speed.
%%%%%%%%%%%%%%%%%%%%%%%%%%%%%%%%%%%%%%%%%%%%%%%%%
%%%%%%%%%%%%%%%%%%%%%%%%%%%%%%%%%%%%%%%%%%%%%%%%%
\section*{Appendix~A: Proof of Proposition~\ref{approxgain}}
In the considered system, when the pinching antenna is located at horizontal coordinate $x$, the instantaneous achievable rate of user $n$ can be expressed as follows:
\begin{equation}
r_n(x)=\log_2\left(1+\frac{\rho\eta^2}{(x-x_n)^2+y_n^2+d^2}\right).
\end{equation}
Since $r_n(x)$ is maximized at $x=x_n$, the first-order derivative with respect to $x$ satisfies
\begin{equation}
r_n'(x_n)=0.
\end{equation}
Moreover, the second-order derivative at $x=x_n$ is given by
\begin{equation}
r_n''(x_n)=-\frac{2\rho\eta^2}{\ln 2\left(y_n^2+d^2\right)\left(y_n^2+d^2+\rho\eta^2\right)}.
\end{equation}
By applying the second-order Taylor expansion around $x=x_n$, it can be obtained that
\begin{align}
r_n(x)&\approx r_n(x_n)\!+\!r_n'(x_n)(x\!-\!x_n)\!+\!\frac{1}{2}r_n''(x_n)(x\!-\!x_n)^2\nonumber\\
&=r_n(x_n)\!-\!\kappa_n(x\!-\!x_n)^2,
\end{align}
where
\begin{equation}
\kappa_n=-\frac{1}{2}r_n''(x_n)=\frac{\rho\eta^2}{\ln 2\left(y_n^2+d^2\right)\left(y_n^2+d^2+\rho\eta^2\right)}.
\end{equation}
Therefore, with antenna placement, the instantaneous achievable rate at the service point can be approximated as
\begin{equation}
r_n(x_n^\mathrm{pt})\approx r_n(x_n)-\kappa_n\left(x_n^\mathrm{pt}-x_n\right)^2.
\end{equation}
Under antenna roaming, the spatial average of $r_n(x)$ over the service interval of user $n$ can be approximated as follows:
\begin{align}
\bar{r}_n^\mathrm{AR}&=\frac{1}{x_{n+1}^\mathrm{bd}\!-\!x_n^\mathrm{bd}}\int_{x_n^\mathrm{bd}}^{x_{n+1}^\mathrm{bd}}\!r_n(x)dx\nonumber\\
&\approx\frac{1}{x_{n+1}^\mathrm{bd}\!-\!x_n^\mathrm{bd}}\int_{x_n^\mathrm{bd}}^{x_{n+1}^\mathrm{bd}}\!\left[r_n(x_n)\!-\!\kappa_n(x\!-\!x_n)^2\right]dx\nonumber\\
&=r_n(x_n)\!-\!\kappa_n\frac{1}{x_{n+1}^\mathrm{bd}\!-\!x_n^\mathrm{bd}}\int_{x_n^\mathrm{bd}}^{x_{n+1}^\mathrm{bd}}\!(x\!-\!x_n)^2dx.
\end{align}
For user $n$, $x_n^\mathrm{bd}$ and $x_{n+1}^\mathrm{bd}$ denote the starting and ending points of its antenna roaming service interval, respectively. The center of this interval is defined as
\begin{equation}
x_n^\mathrm{cen}=\frac{x_n^\mathrm{bd}+x_{n+1}^\mathrm{bd}}{2}.
\end{equation}
Let $z=x-x_n^\mathrm{cen}$. Then, $x-x_n=z+x_n^\mathrm{cen}-x_n$, and the integration interval is transformed into $z\in[-a,a]$, where $a=(x_{n+1}^\mathrm{bd}-x_n^\mathrm{bd})/2$. It follows that
\begin{align}
&\frac{1}{x_{n+1}^\mathrm{bd}\!-\!x_n^\mathrm{bd}}\!\int_{x_n^\mathrm{bd}}^{x_{n+1}^\mathrm{bd}}\!(x\!-\!x_n)^2dx\nonumber\\
=&\frac{1}{2a}\int_{-a}^{a}\!\left(z\!+\!x_n^\mathrm{cen}\!-\!x_n\right)^2dz\nonumber\\
=&\frac{1}{2a}\int_{-a}^{a}\!\left[z^2\!+\!2z\left(x_n^\mathrm{cen}\!-\!x_n\right)\!+\!\left(x_n^\mathrm{cen}\!-\!x_n\right)^2\right]dz\nonumber\\
=&\frac{1}{2a}\int_{-a}^{a}\!z^2dz\!+\!\frac{x_n^\mathrm{cen}\!-\!x_n}{a}\!\int_{-a}^{a}\!zdz\!+\!\left(x_n^\mathrm{cen}\!-\!x_n\right)^2.
\end{align}
Since the integration interval is symmetric around zero, one has $\int_{-a}^{a}zdz=0$. Moreover,
\begin{equation}
\frac{1}{2a}\int_{-a}^{a}z^2dz=\frac{1}{2a}\cdot\frac{2a^3}{3}=\frac{a^2}{3}=\frac{\left(x_{n+1}^\mathrm{bd}-x_n^\mathrm{bd}\right)^2}{12}.
\end{equation}
Thus, the above expression can be simplified as follows:
\begin{equation}
\frac{1}{x_{n+1}^\mathrm{bd}\!-\!x_n^\mathrm{bd}}\!\int_{x_n^\mathrm{bd}}^{x_{n+1}^\mathrm{bd}}\!(x-x_n)^2dx=\left(x_n^\mathrm{cen}\!-\!x_n\right)^2+\frac{\left(x_{n+1}^\mathrm{bd}\!-\!x_n^\mathrm{bd}\right)^2}{12}.
\end{equation}
By substituting the approximations of $r_n(x_n^\mathrm{pt})$ and $\bar{r}_n^\mathrm{AR}$ into the sum-rate gain expression in Proposition~\ref{gain}, and rearranging the terms, the approximated sum-rate gain is obtained as
\begin{align}
\Delta R_\mathrm{sum}\approx &\sum_{n=1}^{N}\!\left(\frac{x_{n+1}^\mathrm{bd}\!-\!x_n^\mathrm{bd}}{D_x}\!-\!\frac{1}{N}\!+\!\frac{x_n^\mathrm{pt}\!-\!x_{n-1}^\mathrm{pt}}{Tv_\mathrm{max}}\right)r_n(x_n)\nonumber\\
&+\!\sum_{n=1}^{N}\!\left(\frac{1}{N}\!-\!\frac{x_n^\mathrm{pt}\!-\!x_{n-1}^\mathrm{pt}}{Tv_\mathrm{max}}\right)\kappa_n\left(x_n^\mathrm{pt}\!-\!x_n\right)^2\\
&-\!\sum_{n=1}^{N}\!\frac{x_{n+1}^\mathrm{bd}\!\!-\!x_n^\mathrm{bd}}{D_x}\kappa_n\!\!\left[\left(x_n^\mathrm{cen}\!-\!x_n\right)^2\!\!+\!\frac{\left(x_{n+1}^\mathrm{bd}\!\!-\!x_n^\mathrm{bd}\right)^2}{12}\right],\nonumber
\end{align}
which completes the proof.\QEDA
%%%%%%%%%%%%%%%%%%%%%%%%%%%%%%%%%%%%%%%%%%%%%%%%%
%%%%%%%%%%%%%%%%%%%%%%%%%%%%%%%%%%%%%%%%%%%%%%%%%
\bibliographystyle{IEEEtran}
\bibliography{KaidisBib}
%%%%%%%%%%%%%%%%%%%%%%%%%%%%%%%%%%%%%%%%%%%%%%%%%
%%%%%%%%%%%%%%%%%%%%%%%%%%%%%%%%%%%%%%%%%%%%%%%%%
\end{document}